	\documentclass[
		paper=a4paper,%
		11pt,american,%
		DIV=12,headinclude=true,headlines=0,footinclude=true,footlines=-5,twoside=semi%
	]{scrartcl}
	\def\version{arxiv}
	\usepackage{hyperref}

\input{preamble}
\hypersetup{hidelinks}

\def\mytitle{Efficient Second-Order Shape-Constrained Function Fitting}

	\title{%
		\mytitle%
		\thanks{%
			The first author is supported in part by 
			National Science Foundation Grant 1718533.
			The last author is supported by the 
			Natural Sciences and Engineering Research Council of Canada 
			and the Canada Research Chairs Programme.
		}%
	}
	\author{
		David Durfee%
			\thanks{%
				Georgia Institute of Technology \;$\cdot$\;
				\texttt{\{ddurfee,ygao380\}\,@\,gatech.edu}
			}
		\and
		Yu Gao${}^{\fnsymbol{footnote}}$
		\and
		Anup B. Rao%
			\thanks{%
				Adobe Research \;$\cdot$\;
				\texttt{anuprao\,@\,adobe.com}
			}
		\and
			Sebastian Wild%
			\thanks{%
				University of Waterloo \;$\cdot$\;
				\texttt{wild\,@\,uwaterloo.ca}
			}%
	}

\begin{document}

\maketitle

\begin{abstract}
We give an algorithm to compute a one-dimensional shape-constrained function that best fits 
given data in weighted-$L_{\infty}$ norm. We give a \emph{single} algorithm that works for
a variety of commonly studied shape constraints including monotonicity, Lipschitz-continuity 
and convexity, and more generally, any shape constraint expressible by bounds on first- and/or 
second-order differences.
Our algorithm computes an
approximation with additive error $\epsilon$ 
in $O\left(n \log \frac{U}{\epsilon} \right)$ time, where $U$
captures the range of input values. 
We also give a simple greedy algorithm that runs in
$O(n)$ time for the special case of unweighted $L_{\infty}$ convex regression.
These are the first (near-)linear-time algorithms for second-order-constrained function fitting.
To achieve these results,
we use a novel geometric interpretation of the underlying dynamic programming problem.
We further show that a generalization of the corresponding problems to directed acyclic
graphs (DAGs) is as difficult as linear programming.
\end{abstract}

\pagenumbering{arabic}

\section{Introduction}
\label{sec:intro}

We consider the fundamental problem of finding a function $f$ that approximates 
a given set of data points $(x_1,y_1),\ldots,(x_n,y_n)$ in the plane 
with smallest possible error,
\ie, $f(x_i)$ shall be close to $y_i$ (formalized below),
subject to shape constraints on the allowable functions $f$, such as 
being increasing and/or concave.
More specifically, we present a new algorithm that can handle 
arbitrary constraints on the
(discrete) first- and second-order derivatives of $f$.

When we only require $f$ to be weakly increasing, the problem is known as 
isotonic regression, a classic problem in statistics;
(see, \eg, \cite{GroeneboomJongbloed2014} for history and applications).
It has more recently also found uses in 
machine learning~\cite{KalaiSastry2009,KakadeKanadeShamirKalai2011,GantiSastryBalzanoWillet2015}.

In certain applications, further shape restrictions are integral part of the model:
For example, microeconomic theory suggests that production functions are weakly increasing and concave
(modeling diminishing marginal returns);
similar reasoning applies to utility functions.
Restricting $f$ to functions with bounded derivative (Lipschitz-continuous functions)
is desirable to avoid overfitting~\cite{KakadeKanadeShamirKalai2011}.
All these shape restrictions can be expressed by inequalities for first and second derivatives of $f$;
their discretized equivalents are hence amenable to our new method.
Shape restrictions that we cannot directly handle are studied in~\cite{TsourakakisPengTsiarliMillerSchwartz2011} 
($f$ is piecewise constant and the number of breakpoints is to be minimized)
and~\cite{Stout2008} (unimodal $f$).
For a more comprehensive survey of shape-constrained function-fitting problems and their applications,
see~\cite[\S1]{GuntuboyinaSen2018}.
Motivated by these applications, the problems have been
studied in statistics (as a form of nonparametric regression), 
investigating, \eg, their consistency as estimators and their rate of convergence%
~\cite{GroeneboomJongbloed2014,GuntuboyinaSen2018,Balasz2016}.

While fast algorithms for isotonic-regression variants have been designed~\cite{Stout2014},
both~\cite{Lim2018} and~\cite{Bach2018} list shape constraints
beyond monotonicity as important challenges.
For example, fitting (multidimensional) convex functions is 
mostly done via quadratic or linear programming solvers~\cite{MazumderChoudhuryIyengarSen2018}.
In his PhD thesis, Balázs writes that current 
``methods are computationally too expensive for practical use, [so]
their analysis is used for the design of a heuristic 
training algorithm which is empirically evaluated''~\cite[p.\,1]{Balasz2016}.

This lack of efficient algorithms motivated the present work.
Despite a few limitations discussed below 
(implying that we do not yet solve Balázs' problem), 
we give the first \emph{near-linear-time} algorithms
for any function-fitting problem with second-order shape constraints (such as convexity).
We use dynamic programming (DP) with a novel geometric encoding for the ``states''.
Simpler versions of such geometric DP variants were used for isotonic regression~\cite{Rote2019}
and are well-known in the competitive programming community; 
incorporating second-order constraints efficiently is our main innovation.

\paragraph{Problem definition.}
Given the vectors $\vec x = (x_1,\ldots,x_n) \in \mathbb R^n$
and $\vec y \in \mathbb R^n$,
an error norm $d$ and shape constraints (formalized below),
compute $\vec f = (f_1,\ldots,f_n)$ satisfying the shape constraints 
with minimal $d(\vec f,\vec y)$,
\ie, we represent $f$ via its values $f_i = f(x_i)$ at the given points.
$d$ is usually an $L_p$ norm, 
$d(\vec x,\vec y) = \bigl(\sum_i |x_i-y_i|^p\bigr){}^{1/p}$;
least squares ($p=2$) dominate in statistics, but more general
error functions have been studied for isotonic regression~\cite{LussRosset2014,KyngRaoSachdeva2015,Lim2018,Bach2018}.
We will consider the \emph{weighted} $L_\infty$ norm, \ie,
\(
	d(\vec f, \vec y) = \max_{i\in[n]} w_i \mathopen| f_i - y_i \mathclose|
\),
where $[n] = \{1,\ldots,n\}$ and $\vec w \in \mathbb R_{\ge0}^n$ is a given vector of weights.

Since we are dealing with discretized functions (a vector $\vec f$),
restrictions for derivatives $f'$ and $f''$ have to be discretized, as well.
We define local slope and curvature as
\[
	f'_i \wrel= \frac{f_i - f_{i-1}}{x_i-x_{i-1}},\quad (i\in[2..n]),
	\quad\text{and}\quad
	f''_i \wrel= \frac{f'_{i}-f'_{i-1}}{x_i-x_{i-1}},\quad (i\in[3..n]);
\]
the shape constraints are then given in the form of vectors 
$\vec {f^{\prime-}}, \vec{f^{\prime+}}, \vec{f^{\prime\prime-}}, \vec{f^{\prime\prime+}}$
of bounds for the first- and second-order differences, \ie,
we define the set of feasible answers as
\(F = \bigl\{\vec f \in \mathbb R^n \bigm|
	\vec {f^{\prime-}} \le \vec{f'} \le \vec{f^{\prime+}} \bin\wedge 
	\vec {f^{\prime\prime-}} \le \vec{f''} \le \vec{f^{\prime\prime+}} 
\bigr\}\)
where inequalities on vectors mean the inequality on all components.
The \emph{weighted-$L_\infty$ function-fitting problem with second-order shape constraints} 
is then to find 
\begin{align}
\label{eq:main}
	\vec{f^*} \wwrel= \operatornamewithlimits{\arg\min}_{\vec f \in F} \left( \max_{i} \,w_i \cdot |f_i - y_i|\right).
\end{align}
Often, we only need a lower resp.\ upper bound; we can achieve that 
by allowing $-\infty$ and $+\infty$ entries in $f^{\prime\pm}_i$ and $f^{\prime\prime\pm}_i$.
For example, setting $\vec{f^{\prime\prime-}} = 0$, $\vec{f^{\prime-}} = \vec{f^{\prime\prime-}} = +\infty$
and $\vec{f^\prime-} = -\infty$, we can enforce a convex function/vector.
We also consider the decision-version of the problem:
given a bound $L$, decide if there is an $\vec f\in F$ with 
\(\max_{i} w_{i} \left|f_i - y_{i}\right|  \leq L\),
and if so, report one.

\paragraph{Contributions.}
Our main result is a \emph{single} $O(n)$-time algorithm for the decision problem
of function fitting with second-order constraints;
see Theorem~\ref{thm:main} for the precise statement. 
With binary search, this readily yields
an additive $\epsilon$-approximation for \eqref{eq:main}, 
and thus weighted $L_{\infty}$ isotonic regression, 
convex regression and Lipschitz convex regression, in 
$O\bigl( n \log \frac{U}{\epsilon} \bigr)$ time (Theorem~\ref{thm:main1}), 
where $U = (\max_i w_i)\cdot (\max_i y_i - \min_i y_i )$. 
\ifconf{%
	In the extended online version of this article (\href{https://arxiv.org/abs/1905.02149}{\texttt{arxiv.org/abs/1905.02149}}), 
	we give a simple greedy algorithm
}{%
	In the appendix, we give a simple greedy algorithm (see Theorem~\ref{thm:minor}) 
}%
for \emph{unweighted} ($\vec w=1$) $L_{\infty}$ convex regression that runs in $O(n)$ time.
Finally, we show that a generalization of the problem to DAGs 
(where the applied first- and second-order difference constraints are restricted by the graph),
is as hard as linear programming%
\ifconf{.}{, see Appendix~\ref{app:dags-hard}.}

\paragraph{Related work.}

Stout~\cite{Stout2014} surveys algorithms for various versions of isotonic regression;
they achieve near-linear or even linear time for many error metrics.
He also considers the generalization to any partial order 
(instead of the total order corresponding to weakly increasing functions).
A related task is to fit a piecewise-constant function (with a prescribed number
of jumps) to given data. \cite{FournierVigneron2009,FournierVigneron2013} solve
this problem for $L_\infty$ in optimal $O(n\log n)$ time.
Since the geometric constraints are much easier than in our case,
a simple greedy algorithm suffices to solve the decision version.

For more restricted shapes, less is known.
\cite{Stout2008} gives a $O(n \log n)$ solution for unimodal regression.
\cite{AgarwalPhillipsSadri2010}~gives an $O(n\log n)$ algorithm for unweighted 
$L_2$ Lipschitz isotonic regression
and a $O(n \operatorname{poly}(\log n))$ time algorithm for Lipschitz unimodal regression.
\cite{MazumderChoudhuryIyengarSen2018} describes (multidimensional) 
$L_2$ convex regression algorithms based quadratic programming.
Fefferman~\cite{Fefferman2012} studied a closely related problem of smooth
interpolation of data in Euclidean space minimizing a certain norm defined on the
derivatives of the function. 
His setup is much more general, but his algorithm cannot find arbitrarily 
good interpolations ($\epsilon$ is fixed for the algorithm).
All fast algorithms above consider classes defined
by constraints on the \emph{first} derivative only, 
not the second derivative as needed for convexity.
To our knowledge, the fastest prior solution for any convex regression problem
is solving a linear program, which will imply super-linear time.

We use a geometric interpretation of dynamic-programming states
and represent them implicitly.
The work closest in spirit to ours is a recent article by Rote~\cite{Rote2019};
establishing the transformation of states is much more complicated in the presently studied problem, though.
Implicitly representing a series of more complicated objects using
data structures has been used in geometric
and graph algorithms, such as multiple-source shortest paths~\cite{Klein2005}
and shortest paths in polygons~\cite{Chazelle1982,LeePreparata1984,Erickson2009}.
The only other work (we know of) that interprets dynamic programming
geometrically is~\cite{TsourakakisPengTsiarliMillerSchwartz2011}.

There is a rich literature on methods for speeding up dynamic
programming~\cite{Yao1980,Yao1982,EppsteinGalilGiancarlo1988,GalilGiancarlo1989}.
They involve a variety of powerful techniques
such as monotonicity of transition points,
quadrangle inequalities,
and Monge matrix searching~\cite{AggarwalKlaweMoranShorWilber1987},
many of which have found applications in other settings.
The focus of these methods is to reduce the (average)
number of transitions that a state is involved in,
often from $O(n)$ to $O(1)$.
Therefore, their running times are lower bounded by
the number of states in the dynamic programs.

\subsection{Results}
We formally state our theorem for the decision problem here;
results for shape-constrained function fitting are obtained as corollaries.
For our algorithm, the discrete derivatives (as defined above) are inconvenient
because they involve the $x$-distance between points.
We therefore \emph{normalize} all $x$-distances to $1$ (s.\,t.\ $x_i = i$);
for the second-order constraints, this normalization makes the introduction of 
an additional parameter necessary, the scaling factors $\alpha_i$ (see below).
\begin{definition}[1st/2nd-diff-constrained vectors]
\label{def:problem}
	Let $n$-dimensional vectors 
	$\vec x^- \leq \vec x^+$ (value bounds), 
	$\vec y^- \leq \vec y^+$ (difference bounds), 
	$\vec z^- \leq \vec z^+$ (second-order difference bounds), and 
	$\vec \alpha> 0$ be given. 
	We define $\solutionset \subset \mathbb R^n$ to be the set of all 
	$\vec b\in\mathbb R^n$ that satisfy the following constraints:
	\begin{align*}
			\forall i\in[1..n]
			\rel{} 
			x_i^- &\rel\leq b_i \rel\leq x_i^+
		&
			\text{(value constraints)}
	\\
			\forall i\in[2..n]
			\rel{} 
			y_i^- &\rel\leq b_i - b_{i-1} \rel\leq y_i^+
		&
			\text{(first-order constr.)}
	\\
			\forall i\in[3..n]
			\rel{} 
			z_i^- &\rel\leq (b_i - b_{i-1}) - \alpha_i(b_{i-1} - b_{i-2}) \rel\leq z_i^+
		&
			\text{(second-order constr.)}
	\end{align*}
	Moreover, we consider the ``truncated problems'' $\solutionsettrunc{k}$, 
	where $\solutionsettrunc{k}$ is the
	set of all $\vec b\in\mathbb R^n$ that satisfy the constraints up to $k$ (instead of $n$).
\end{definition} 
A visualization of an example is shown in Figure~\ref{fig:example}.
We can encode an instance 
$(\vec x, \vec y, \vec {f^{\prime\pm}}, \vec{f^{\prime\prime\pm}})$ 
of the decision version of the weighted-$L_\infty$
function-fitting problem with second-order constraints
as 1st/2nd-diff-constrained vectors by setting
\begin{align*}
x^{\pm}_i &\wrel= y_i \pm L / w_i, &
y^{\pm}_i &\wrel= f^{\prime\pm} \cdot (x_i-x_{i-1}), \\
z^\pm_i &\wrel= f^{\prime\prime\pm}\cdot (x_i-x_{i-1})^2, &
\alpha_i &\wrel= \frac{x_i-x_{i-1} } { x_{i-1}-x_{i-2} }.
\end{align*}
So, our goal is to efficiently compute some $\vec b \in \solutionset$ or determine 
that $\solutionset = \emptyset$.
Our core technical result is a linear-time algorithm for this task:

\begin{figure}[tbhp]
\plaincenter{%
\ifconf{\scalebox{1.25}}{\resizebox{\linewidth}!}{%
\includegraphics{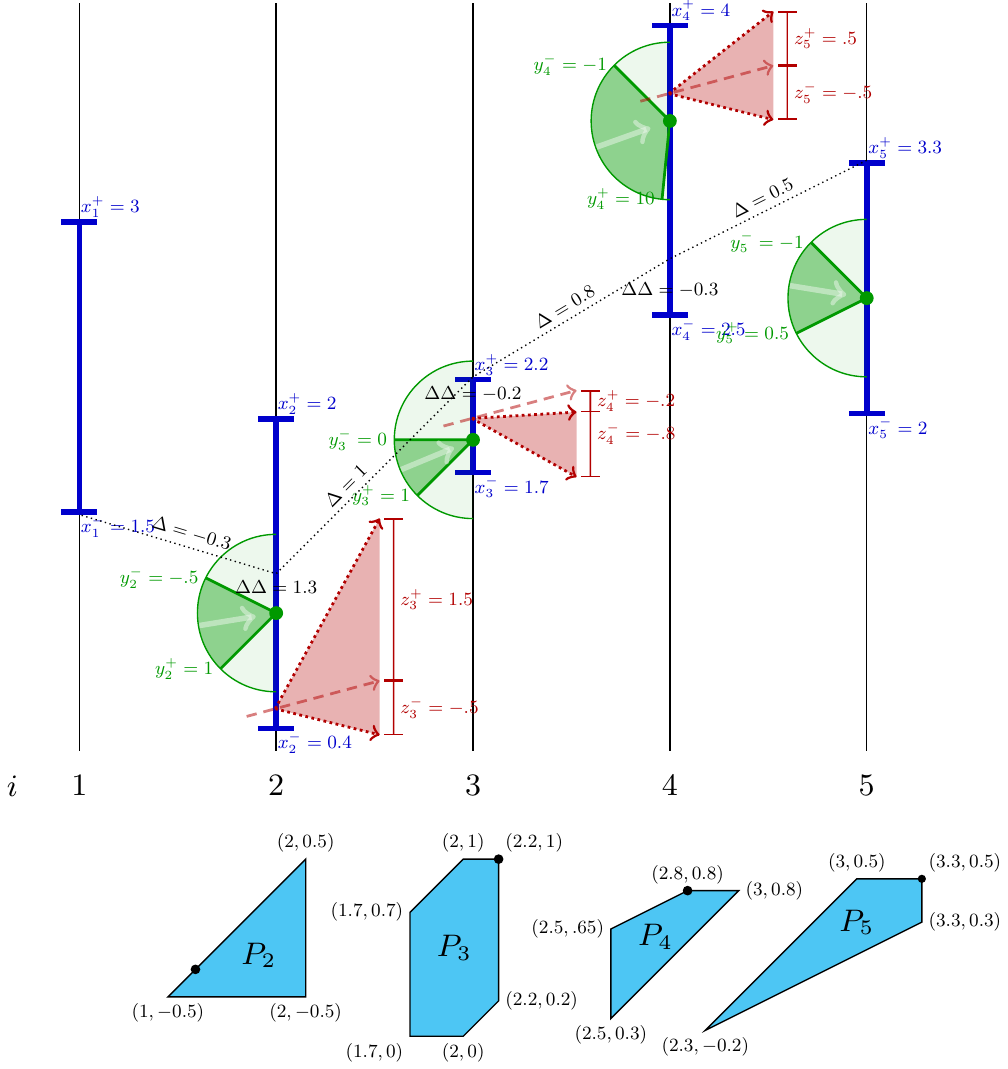}
}
}
\ifconf{\vspace*{-3ex}}{}
\caption{%
	Exemplary input for the 1st/2nd-diff-constrained decision problem
	with $\vec \alpha=1$.
	Value constraints are illustrated as blue bars. First-order constraints are shown
	as green circles, indicating the allowable incoming angles/slopes; the green dot and the circle
	can be moved up and down within the blue range.
	Finally, second-order constraints are given as red triangles, in which the minimal and maximal
	allowable change in slope is shown (dotted red), based off an exemplary incoming slope (dashed red).
	The thin dotted line shows $\vec b = (1.7,1.2,2.2,2.8,3.3) \in \solutionset$.\protect\\
	Below the visualization of the instance, we show the set of 
	pairs $(b_i,b_i-b_{i-1})$ for $\vec b\in \solutionset_i$, \ie, 
	solutions of the truncated problem; the specific solution is shown as a dot.
	These sets are the \emph{feasibility polygons} $P_i$ (defined in Section~\ref{sec:overview})
	that play a vital role in our algorithm.
	Given all $P_i$, one can easily construct a solution backwards, 
	starting from any point in $P_5$.%
}
\label{fig:example}
\end{figure}

\begin{theorem}[1st/2nd-diff-constrained decision]
\label{thm:main}
	With the notation of Definition~\ref{def:problem},
	in $O(n)$ time, we can compute $\vec b \in \solutionset$ 
	or determine that $\solutionset=\emptyset$.
\end{theorem}
Section~\ref{sec:dyn_prog} will be devoted to the proof.
To simplify the presentation, 
we will assume throughout that 
$\vec{x^+}$, $\vec{x^-}$, $\vec{y^+}$, $\vec{y^-}$, $\vec{z^+}$, $\vec{z^-}$
are bounded.%
\footnote{%
	Some problems are stated with $\pm\infty$ values, but 
	we can always replace unbounded values in the algorithms
	with an (input-specific) sufficiently large finite number.
}
For the optimization version of the problem, Equation~\eqref{eq:main},
we consider approximate solutions in the following sense.
\begin{definition}[$\epsilon$-approximation]
\label{defn:approx}
We call $\vec f\in F$ an
$\epsilon$-approximate solution to the weighted $L_{\infty}$ function-fitting problem if
it satisfies
\[  
		\max_{i} w_{i}\left|f_i -y_{i}\right|  
	\wwrel\leq 
		\min_{\vec g\in F} \left(
			\max_{i}w_{i}\left|g_i-y_{i}\right| 
		\right) 
		+ \epsilon.
\]
\end{definition}
By a simple binary search on $L$, we can find approximate solutions.
\begin{theorem}[Main result]
\label{thm:main1}
There exists an algorithm that computes an $\epsilon$-approximate solution to the 
weighted-$L_\infty$ convex regression problem that runs in $O(n\log \frac{U}{\epsilon})$ time,
where $U = (\max_i w_i)(\max_i y_i - \min_i y_i ).$ The same holds true for isotonic
regression, Lipschitz isotonic regression, convex isotonic regression.
\end{theorem}

\FullVersion{
\begin{proof}
We will argue for the case of convex regression here, other cases are similar.
Abbreviate $L(\vec f) = \max_i w_i | f_i - y_i|$.
For a given $L$, the decision version of convex
regression can be solved in $O(n)$ time using Theorem~\ref{thm:main}.  That is, in
$O(n)$ time, we can either find $\vec f \in F$ such that $L(\vec f) \leq  L$
or conclude that for all $\vec f \in F,$ $L(\vec f) >  L.$  If we know an
$L_0$ for which there exists $\vec f \in F$ with $L(\vec f) \leq  L_0$,
then we can do a binary search for $L_c$ in $[0,L_0].$  
We can easily find such an $L_0$ for the convex case:
Let $\vec f = \min y_j$ be constant (hence convex). 
For this $\vec f$, we have 
$L(\vec f) \leq (\max_j w_j)(\max_j y_j - \min_j y_j).$ Therefore, we can take
$L_0 = (\max_j w_j)(\max_j y_j - \min_j y_j)$ and the result immediately follows.
\qed\end{proof}
}
We note that for the specific case of \emph{unweighted} convex function fitting,
there is a simpler linear-time greedy algorithm; 
we give more details on that in 
\ifconf{%
	the extended online version.
}{%
	Appendix~\ref{app:greedy-proof}.
}%
This algorithm was the initial motivation for studying this problem and for 
the geometric approach we use.
For more general settings, in particular second-order differences that are allowed to be
both positive and negative, the greedy approach does not work; 
our generic algorithm, however, is almost as simple and efficient.

\section{First- and second-order difference-constrained vectors}
\label{sec:dyn_prog}

In this section, we present our main algorithm and prove Theorem~\ref{thm:main}. 
In Section~\ref{sec:overview}, we give an overview and introduce the feasibility polygons $P_i$.
Section~\ref{subsec:transform} shows how $P_i$ can be inductively computed from $P_{i-1}$
via a geometric transformation.
We finally show how this transformation can be computed efficiently,
culminating in the proof of Theorem~\ref{thm:main}, in Section~\ref{sec:algorithm}.
Two proofs are deferred to 
\ifconf{%
	the extended online version of this article.
}{%
	Appendix~\ref{app:polytope-shift} and~\ref{app:backtracing}.
}

\subsection{Overview of the algorithm}
\label{sec:overview}

Recall that the problem we want to solve, in order to prove Theorem~\ref{thm:main}, is
finding a feasible point $\vec b$ in $\solutionset$ from Definition~\ref{def:problem}. 
Our algorithm will use dynamic programming (DP) where each state is associated with the
feasible $b_i$ in the truncated problem. 
We will iteratively determine all
$b_i$ such that $b_i$ is the $i$th entry of some $\vec b \in \solutionsettrunc{i}$. 

Feasible $b_i$ have to respect the first- and second-order difference constraints.
To check those, we also need to know
the possible pairs $(b_{i-1},b_{i-2})$ of $(i-1)$th
and $(i-2)$th entries for some $\vec b \in \solutionsettrunc{i-1}$,
so the states have to maintain more information than the $b_i$ alone.
It will be instrumental to \emph{rewrite} this pair as $(b_{i-1},b_{i-1} - b_{i-2})$, 
the combination of valid values $b_{i-1}$ and valid \emph{slopes} 
at which we entered $b_{i-1}$
for a solution in $\solutionsettrunc{i-1}$.
From that, we can determine the valid slopes at which we can \emph{leave} $b_{i-1}$ using our
shape constraints. 
We thus define the \emph{feasibility polygons}
\begin{align}
	\label{def:each_polytope}
		P_i 
	\wwrel= 
		\bigl\{(x,y) \bigm| \exists \vec b \in \solutionsettrunc{i} : 
			x = b_i \wedge y = b_i - b_{i-1} \bigr\}
\end{align}
for $i=2,\ldots,n$.
See Figure~\ref{fig:example} for an example.
We view each point in $P_i$ as a ``state'' in our DP algorithm, and our goal becomes
to efficiently compute $P_i$ from $P_{i-1}$. 
The key observation is that each $P_i$ is indeed an $O(n)$-vertex
convex polygon,
and we only need an efficient way to compute the \emph{vertices} of $P_{i}$ from those of $P_{i-1}$.
This needs a clever representation, though, 
since all vertices can change when going from $P_{i-1}$ to $P_i$.
A closer look reveals that we can represent the
vertex transformations \emph{implicitly}, without actually updating each vertex,
and we can combine subsequent transformations into a single one.
More specifically, if we consider the boundary of $P_{i-1}$, 
the transformation to $P_i$ consists of two steps: 
(1) a linear transformation for the upper and lower hull of
$P_{i-1}$, and (2) a truncation of the resulting polygon by vertical and horizontal lines 
(\ie, an intersection of the polygon and a half-plane).

The first step requires a more involved proof and uses that all line segments of $P_i$
have weakly positive slope (``$\polytope$'', formally defined below).
Implicitly computing the first transformation as we move between $P_i$ is
straightforward, only requiring a composition of linear operations 
(a different one, though, for upper and lower hull).
We can apply the cumulative transformation whenever we need
to access a vertex.

The second step is conceptually simpler, but more difficult to implement efficiently, 
as we have to determine where a line cuts the polygon in amortized constant time.
For this operation, we separately store the vertices of the upper and lower hull of $P_i$ 
in two arrays, sorted by increasing $x$-coordinate;
since $P_i$ is $\polytope$, $y$-values are also increasing.
A linear search for intersections has overall $O(n)$ cost 
since we can charge individual searches to deleted vertices.

Finally, if $P_n \ne \emptyset$, 
we compute a feasible vector $\vec b$ backwards, 
starting from any point in $P_n$. 
Since we do not explicitly store the $P_i$, this requires
successively ``undoing'' all operations (going from $P_i$ back to $P_{i-1}$);
see 
\ifconf{the extended version of this paper}{Appendix~\ref{app:backtracing}}
for details.

\subsection{Transformation from state $P_{i-1}$ to $P_i$}\label{subsec:transform}

We first define the structural property ``$\polytope$'' that
our method relies on.
\begin{definition}[$\polytope$]
\label{def:our_polytope}
	We say a polygon $P\subseteq\mathbb R^2$ with vertices
	$v_1,\ldots,v_k$ is $\polytope$ if $\mathrm{slope}(v_i,v_j) \geq 0$
	for all edges $(v_i,v_j)$ of $P$.
	Here, the slope between two points $v_1 = (x_1,y_1)$, $v_2 = (x_2,y_2) \in \mathbb{R}^2$ 
	is defined as $\mathrm{slope}(v_1,v_2) = \frac{y_2 - y_1}{x_2 - x_1}$, when $x_1 \ne x_2$,
	and $\mathrm{slope}(v_1,v_2) = \infty$, otherwise.
\end{definition}
We will now show that $P_i$ can be computed by applying
a simple geometric transformation to $P_{i-1}$.
In passing, we will prove (by induction on $i$) that all $P_i$ are $\polytope$.
For the base case, note that 
\(
	P_2 = \{(b_2,b_2-b_1) \mid 
			x_{1}^- \le b_1 \le x_{1}^+ \wedge
			x_{2}^- \le b_2 \le x_{2}^+ \wedge 
			y_{2}^- \le b_2-b_1 \le y_{2}^+
		\}
\), which 
is an intersection of $6$ half-planes. The slopes of the defining inequalities 
are all non-negative or infinite, so $P_2$ is $\polytope$.

Let us now assume that $P_{i-1}$, $i\ge3$, is $\polytope$; 
we will consider the transformation
from $P_{i-1}$ to $P_i$ and show that it preserves this property.
We begin by separating the transformation from $P_{i-1}$ to $P_i$ 
into two main steps.

\paragraph{Step 1: Second-order constraint only.}
\label{sec:step1}
For the first step, we ignore the value and first-order constraints
at index $i$.
This will yield a convex polygon, $\sndOrdOnlyP $, that contains $P_i$;
in \hyperref[sec:step2]{Step 2}{}, we will add the other constraints at $i$ to obtain $P_i$ itself.
\begin{definition}[$\sndOrdOnlyP $: 2nd-order-only polygons]
\label{def:polytope_with_no_x_or_y}
	For a fixed $i$, consider the modified problem with $x_i^-,y_i^- = -\infty$ and
	$x_i^+,y_i^+ = \infty$. Define the second-order-only polygon, $\sndOrdOnlyP $, as the polygon $P_i$ of this modified
	problem (considering only the $z_i$ constraints at $i$).
\end{definition}
The statement of the following lemma is very simple observation, 
but allows us to compute $\sndOrdOnlyP$ from $P_{i-1}$ with an explicit 
geometric construction, (whereas such seemed not obvious 
for the original feasibility polygons).
\begin{figure}[tbh]
\plaincenter{
\ifconf{\scalebox{1.1}}{\resizebox{\linewidth}!}{%
\includegraphics{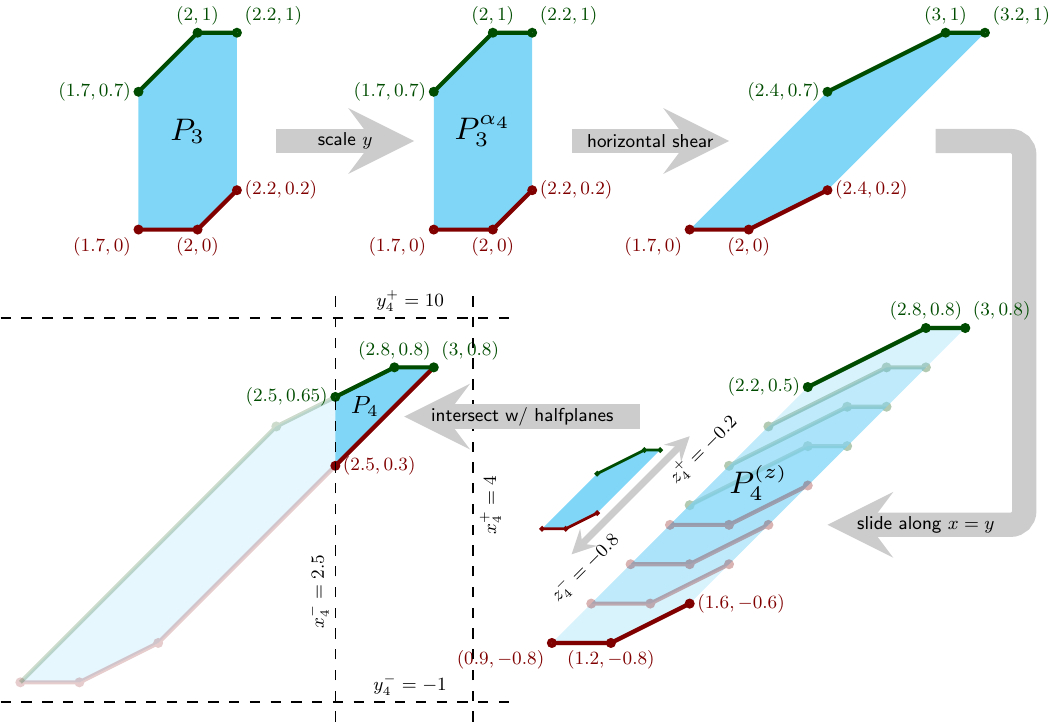}
}
}
\caption{%
	The transformation from $P_3$ to $P_4$ for the example instance of Figure~\ref{fig:example}.
	Upper and lower hull are shown separately in green resp.\ red.
}
\label{fig:P3-to-P4}
\end{figure}

\begin{lemma}[$\sndOrdOnlyP $: scaled, sheared and shifted $P_{i-1}$]
\label{lem:stretch}~\\
	$\sndOrdOnlyP  = \bigl\{
		(x+\alpha_i y + z, \alpha_i y + z) \mid (x,y) \in P_{i-1}, z \in [z_i^-,z_i^+] 
		\bigr\}$.
\end{lemma}
\begin{proof}
	The only constraint at $i$ is 
	$z_i^- \leq (b_i - b_{i-1}) - \alpha_i(b_{i-1} - b_{i-2}) \leq z_i^+$.
	We rewrite this as (a) a constraint
	for $b_i - b_{i-1}$, using that  $b_{i-1} - b_{i-2}$ is the $y$-coordinate in $P_{i-1}$, 
	and (b) a constraint for $b_i$, using 
	that, additionally, $b_{i-1}$ is the $x$-coordinate in $P_{i-1}$.
\qed\end{proof}
Once we have computed this polygon $\sndOrdOnlyP $, computing $P_i$ is easy:
adding the constraints $x_i^- \leq x \leq x_i^+$ and  $y_i^- \leq y \leq y_i^+$ 
requires only cutting $\sndOrdOnlyP $ with two horizontal and vertical lines.
We give a visual representation of the mapping on an example in Figure~\ref{fig:P3-to-P4}.
We break the above mapping into two simpler stages:
\begin{corollary}[$\sndOrdOnlyP $: sheared and shifted $P_{i-1}^{\alpha_i}$]
\label{cor:scaling}%
~\\
	Setting $P_{i-1}^{\alpha_i} = \{(x, \alpha_i y) \mid (x,y) \in P_{i-1}\}$, we have\\
	$\sndOrdOnlyP  = \bigl\{(x+y + z,y + z) \mid (x,y) \in P_{i-1}^{\alpha_i}, z \in [z_i^-,z_i^+] \bigr\}$.
\end{corollary}
We note that scaling the $y$-coordinate by $\alpha_i$ preserves the $\polytope$-property:
\begin{lemma}%
\label{lem:scaling}
	Let $\alpha \ge 0$.
	If $P$ is $\polytope$, so is $P^\alpha = \{(x,\alpha y)\mid (x,y \in P) \}$.
\end{lemma}
\begin{proof}
	Scaling the $y$-coordinates will preserve all of the vertices of $P$, and also scale the
	slope of each vertex pair by $\alpha\ge 0$. So, $P^\alpha$ is $\polytope$.
\qed\end{proof}
That leaves us with the core of the transformation, from $P_{i-1}^{\alpha_i}$ to
$\sndOrdOnlyP $. Intuitively, it can be viewed as sliding $P_{i-1}^{\alpha_i}$
along the line $x=y$ by any amount $z \in [z_i^-, z_i^+]$ and taking the union thereof,
(see Figure~\ref{fig:P3-to-P4}).
To compute the result of this operation, we split the boundary into upper and lower hull.
\begin{definition}[Upper/lower hull]
\label{def:upper}
	Let $P$ be a convex polygon with vertex set $V$. 
	We define the upper hull (vertices) resp.\ lower hull (vertices) of $P$ as 
	\begin{align*}
		\upperhull(P) &\wrel= \bigl\{u_i = (x_i,y_i) \in V \bigm|  \nexists (x_i,y) \in P : y > y_i \bigr\}
	\\
		\lowerhull(P) &\wrel= \bigl\{u_i = (x_i,y_i) \in V \bigm|  \nexists (x_i,y) \in P : y < y_i \bigr\}
	\end{align*}
	Unless specified otherwise, hull vertices are ordered by increasing $x$-coordinate.
\end{definition}
Note that a vertex can be in both hulls. 
Moreover, the leftmost vertices in $\upperhull(P)$ and $\lowerhull(P)$ always have the same $x$-coordinate,
similarly for the rightmost vertices.
As proved in Lemma~\ref{lem:stretch}, each point in $P_{i-1}^{\alpha_i}$ is mapped to a
line-segment with slope $1$; we give this mapping a name.
\begin{definition}[2nd-order $P$ transform]
	Let $f_i((x,y))$ be the line-segment $\{(x + y + z, y + z)\mid  z \in [z_i^-,z_i^+]\}$ and 
	denote by 
	$f_i^-((x,y)) = (x+y+z_i^-,y+z_i^-)$ and 
	$f_i^+((x,y)) = (x+y+z_i^+,y+z_i^+)$ 
	the two endpoints of $f_i((x,y))$.

	We write $f(S) = \Union_{(x,y)\in S} f((x,y))$ for the element-wise application of $f$
	to a set $S$ of points.
\end{definition} 
The vertices of $\sndOrdOnlyP$ result from transforming the upper hull of $P_{i-1}^{\alpha_i}$ 
by $f_i^+$ and the lower hull by $f_i^-$.
The next lemma formally establishes that applying
$f_i^+$ resp.\ $f_i^-$ to the hulls of $P_{i-1}^{\alpha_i}$
correctly computes $\sndOrdOnlyP $, (again, compare Figure~\ref{fig:P3-to-P4}).%
\begin{lemma}[From $P_{i-1}^{\alpha_i}$ to $\sndOrdOnlyP $ via hulls]
\label{lem:polytope_shift}%
	If $P_{i-1}^{\alpha_i}$ is $\polytope$, then $\sndOrdOnlyP $ is $\polytope$ and
	$\upperhull(\sndOrdOnlyP ) = \{ f_i^-(v_{\mathit{ll}}) \} \cup 
	f_{i}^+(\upperhull(P_{i-1}^{\alpha_i}))$ and
	$\lowerhull(\sndOrdOnlyP ) = f_{i}^{-}(\lowerhull(P_{i-1})) \cup \{ f_i^{+}(v_{\mathit{ur}})\}$,
	where $v_{\mathit{ll}}$ (\underline lower-\underline left) and 
	$v_{\mathit{ur}}$ (\underline upper-\underline right)
	are the first vertex of $\lowerhull(P_{i-1}^{\alpha_i})$ 
	and
	the last vertex of $\upperhull(P_{i-1}^{\alpha_i})$, respectively.
\end{lemma}
\ifconf{%
	A formal proof is given in the extended version.
}{%
	We defer the formal proof to Appendix~\ref{app:polytope-shift}.
}%
Intuitively, since
each point in $P_{i-1}^{\alpha_i}$ is mapped to a line-segment with slope $1$ in
$\sndOrdOnlyP $, $\sndOrdOnlyP $ is obtained by sliding $P_{i-1}^{\alpha_i}$ along the line
$x=y$. 
\ifconf{}{%
	Note here that we could allow $z_i^- = - \infty$ and/or $z_i^+ = \infty$,
	where the functions $f_i^-,f_i^+$ would instead map to the ray centered at $(x,x+y)$ and
	either pointed upwards or downwards with slope $1$.
}%
The full transformation from $P_{i-1}$ to $\sndOrdOnlyP $ can now be stated as:

\begin{lemma}[$P_{i-1}$ to $\sndOrdOnlyP $]
\label{lem:full_mapping}
	Let $f_i^{*,\alpha_i}$ be the function 
	$f_i^{*,\alpha_i}(x,y) = (x + \alpha_i y + z_i^*,\alpha_i y + z_i^*)$ for 
	$* \in \{-,+\}$. If $P_{i-1}$ is $\polytope$, then
	$\sndOrdOnlyP $ is $\polytope$ with 
	\begin{align*}
	\upperhull(\sndOrdOnlyP ) &\wwrel= \bigl\{f_i^{-,\alpha_i}(v_{\mathit{ll}})\bigr\} \cup f_{i}^{+,\alpha_i}(\upperhull(P_{i-1}))\\
	\lowerhull(\sndOrdOnlyP ) &\wwrel= f_{i}^{-,\alpha_i}(\lowerhull(P_{i-1})) \cup \bigl\{ f_i^{+,\alpha_i}(v_{\mathit{ur}})\bigr\}
	\end{align*}
	with $v_{\mathit{ll}}$ and $v_{\mathit{ur}}$ the lower-left resp.\ upper-right vertex of $P_{i-1}$.
\end{lemma}

\begin{proof}
	This follows immediately from Corollary~\ref{cor:scaling} and
	Lemmas~\ref{lem:scaling} and~\ref{lem:polytope_shift}.
\qed\end{proof}

\paragraph{Step 2: Truncating by value and slope.}
\label{sec:step2}
To complete the transformation, we need to add the constraints $x_i^- \leq b_i \leq
x_i^+$ and $y_i^- \leq b_i - b_{i-1} \leq y_i^+$
to $\sndOrdOnlyP$.
This is equivalent to cutting our polygon with two vertical and horizontal planes. 
The following lemma shows that this
preserves the $\polytope$-property. 
\begin{lemma}[\# new vertices]
	\label{lem:invariant}
	If $P_{i-1}$ is $\polytope$ with $k$ vertices, then $P_i$ is either empty or $\polytope$
	with at most $k + 6$ vertices.
\end{lemma}
It follows that over the course of the algorithm, only
$O(n)$ vertices are added in total. This will be instrumental for analyzing the running time.
\begin{proof}
	We know that $\sndOrdOnlyP $ is $\polytope$, and it follows easily from the
	definition that cutting by horizontal and vertical planes will preserve this property.
	Furthermore, note that cutting a convex polygon
	will increase the total number of vertices by at most one.
	We added at most 2 vertices to $P_{i-1}$ to obtain $\sndOrdOnlyP $. 
	We then cut $\sndOrdOnlyP $ by the inequalities 
	$x \leq x_i^+$, $x \geq x_i^-$, $y \leq y_i^-$, and $y \geq y_i^+$, 
	\ie, two horizontal and vertical planes. Each adds at most one vertex, 
	giving the desired upper bound.
\qed\end{proof}

\subsection{Algorithm}
\label{sec:algorithm}

A direct implementation of the transformation of 
Lemma~\ref{lem:full_mapping} yields a ``brute force'' algorithm
that maintains all vertices of $P_i$ and checks if $P_n$ is empty;
(the running time would be quadratic).
	It works as follows:
	\begin{enumerate}[label=\arabic*.]
	\item{} \textsl{[Init]:}\;
		Compute the vertices of $P_2$.
	\item{} \textsl{[Compute $P_i$]:}\;
		For $i=3,\ldots,n$, do the following:
		\begin{enumerate}[label=\arabic{enumi}.\arabic*.,leftmargin=2em]
		\item
		At step $i$, scale the $y$-coordinate of each vertex by $\alpha_i$.
		\item Apply $f_i^+$ resp.\ $f_i^-$ to each vertex, 
		depending on which hull it is in.
		\item Add the new vertex to $\upperhull$ and $\lowerhull$, as per
		Lemma~\ref{lem:full_mapping}.
		\item Delete all the vertices outside $[x_{i}^-,x_{i}^+]\times [y_{i}^-,y_{i}^+]$ and\\
		add the vertices created by intersecting with $[x_{i}^-,x_{i}^+]\times
		[y_{i}^-,y_{i}^+]$.
		\end{enumerate}
	\item{} \textsl{[Compute $\vec b$]:}\;
		If $P_n\ne\emptyset$, compute $(b_1,\ldots,b_n)$ by backtracing.
	\end{enumerate}
Observe that Lemma \ref{lem:full_mapping} applies the \emph{same} linear function 
(multiplication of $y$-coordinate by $\alpha_{i}$ and $f_i^+$ or $f_i^-$) 
to \emph{all} vertices in $\upperhull$ resp.\ $\lowerhull$.
So, we do not need to modify every vertex each time;
instead, we can store~-- separately for 
$\upperhull$ and $\lowerhull$~--
the \emph{composition} of the linear transformations as a matrix. Whenever we access a vertex, 
we take the unmodified vertex and apply the cumulative transformation in $O(1)$ time.

At each step, after applying the linear transformations, by Lemma \ref{lem:full_mapping}
we also need to copy the leftmost vertex of $\lowerhull$, add it to the left of
$\upperhull$ and copy the rightmost vertex of $\upperhull$ and add it to the right of
$\lowerhull$. To add these vertices, we simply apply the inverse of each respective
cumulative transformation such that all stored vertices 
require the same transformation. This will also take $O(1)$ time.

Since all the slopes of $\sndOrdOnlyP $ are non-negative ($\polytope$)
and we keep vertices sorted by $x$-coordinate,
the truncation by a horizontal or
vertical plane can only remove a prefix or suffix
from $\upperhull$ and $\lowerhull$ of $\sndOrdOnlyP $. 
Depending on the constraint we are adding, 
($x \leq x_i^+$, $x \geq x_i^-$, $y \leq y_i^-$, or $y \geq y_i^+$), 
we start at the rightmost or leftmost vertex of the $\upperhull$
and $\lowerhull$, and continue until we find the intersection with the cutting plane. 
We remove all visited vertices. 

This could take $O(n)$ time in any single iteration, but the total
cost over all iterations is $O(n)$ since we start with $O(1)$ vertices 
and add $O(n)$ vertices throughout the entire procedure
(by Lemma~\ref{lem:invariant}). 
This allows us to use two deques (\underline double-\underline ended \underbar{que}ues), 
represented as arrays,
to store the vertices of $\upperhull$ and $\lowerhull$.
Putting this all together gives the linear time algorithm for the decision problem 
``$\solutionset = \emptyset$?''.

To compute an actual solution when $\solutionset\ne \emptyset$, 
we compute $b_n,\ldots,b_1$, in this order.
From the last $P_n$, we can find a feasible $b_n$ (the $x$-coordinate of any point in $P_n$).
Then, we retrace the steps of our algorithm through specific points in each~$P_i$. 
Since intermediate $P_i$ were only implicitly represented,
we have to recover $P_i$ by ``undoing'' the algorithm's operations in reverse order;
this is possible in overall time $O(n)$
by remembering the operations from the forward phase.
\ifconf{%
	The details on the backtracing step are presented in the extended version of this article.
}{%
	The details on the backtracing step are deferred to Appendix~\ref{app:backtracing},
	where we also present the final algorithm.
}

\section{Conclusion}

In this article,
we presented a linear-time dynamic-programming algorithm to decide
whether there is a vector $\vec b$
that lies (componentwise) between given upper and lower bounds and
additionally satisfies inequalities on its first- and second-order (successive) differences.
This method can be used to approximate weighted-$L_\infty$ shape-restricted
function-fitting problems, where the shape restrictions are given as bounds
on first- and/or second-order differences (local slope and curvature).

This is a first step towards much sought-after efficient methods for more 
general convex regression tasks.
A main limitation of our approach is the restriction to one-dimensional problems.
We show in 
\ifconf{%
	the extended version of this article (\href{https://arxiv.org/abs/1905.02149}{\texttt{arxiv.org/abs/1905.02149}})
}{%
	Appendix~\ref{app:dags-hard}
}%
that a natural extension of the problem studied here
to directed acyclic graphs is already as hard as linear programming, 
leaving little hope for an efficient generic solution.
This is in sharp contrast to isotonic regression, where similar extensions
to arbitrary partial orders do have efficient algorithms (for $L_\infty$)~\cite{Stout2014}.
This might also be bad news for multidimensional regression with second-order constraints, 
since higher dimensions entail, among other complications, a non-total order over the inputs.

A second limitation is the $L_\infty$ error metric, 
which might not be adequate for all applications.
We leave the question whether similarly efficient methods are also possible for other metrics
for future work.
A further extension to study is convex \emph{unimodal} regression; 
here, finding the maximum is part of the fitting problem, and so not directly possible
with our presented method.

\section*{Acknowledgments}
We thank Richard Peng, Sushant Sachdeva, and Danny Sleator for insightful discussions,
and our anonymous referees for further relevant references and insightful comments 
that significantly improved the presentation.

	\bibliographystyle{plainurl}
\bibliography{clean-refs}

\clearpage

	\begin{appendix}
	\section*{Appendix}
		\section{Simple greedy algorithm for convex regression}
\label{app:greedy-proof}

In this appendix, we give details on a simpler algorithm for the special case of
unweighted convex function fitting.

\begin{theorem}
	\label{thm:minor}
	There exists an algorithm for the unweighted $L_\infty$ convex regression 
	that runs in $O(n)$ time. 
\end{theorem}

\begin{proof}
We consider the following problem. Given an $n$-dimensional vector $\vec a$, 
and parameter $\Delta \geq 0$, find a convex vector $\vec b$ 
such that $\norm{b-a}_\infty \leq \Delta$, if such a vector
exists. 

This clearly fits under our parameters of Definition~\ref{def:problem} by setting
$\vec{x^-} = \vec a - \Delta, \vec{x^+} = \vec a + \Delta$, 
both $\vec{y^-}$ and $\vec{y^+}$ to be unbounded, and $\vec{z^-} = 0,
\vec{z^+} = \infty$, along with $\vec \alpha = 1$. 
A binary search on $\Delta$ gives a $O(n\log \frac{U}{\epsilon})$ algorithm.
	
	However, this can also be solved by considering the set of points $(i,a_i + \Delta)$ for
all $i$, and taking the lower hull,%
\footnote{%
	The lower hull of a set of points is the subset of vertices $(x_i,y_i)$ of the 
	convex hull, where $y_i$ is the minimal $y$-coordinate of all
	points with the $x$-coordinate $x_i$ in the convex hull;
	see also Definition~\ref{def:upper}.
}
$H(\Delta)$, such that for each point $(i,h_i)$
in this lower hull we set $b_i = h_i$. We claim that the minimum possible $\Delta$ such
that $b_i\ge a_i-\Delta$ is exactly the answer to this problem. If $(i,a_i+\Delta)$ is a
vertex of the convex hull, $b_i=a_i+\Delta$ is always at least $a_i-\Delta$.
Otherwise, let $(j,a_j + \Delta)$, $(k,a_k + \Delta)$ be two vertices of $H$ such that $j < i
< k$. We have
\[
b_i \wrel\ge a_i-\Delta
\]
\[
\iff a_{j}+\Delta+\frac{a_{k}-a_{j}}{k-j}(i-j) \wrel\ge a_i-\Delta
\]
\[
\iff \Delta \wrel\ge \frac{1}{2}\left(a_i-a_{j}+\frac{i-j}{k-j}(a_{k}-a_{j})\right)
\]
If $\Delta$ violates this for some $i,j,k$, then it is impossible to fit a convex function
through the intervals $[(j,a_j - \Delta),(j,a_j + \Delta)]$, $[(i,a_i - \Delta),(i,a_i +
\Delta)]$, and $[(k,a_k - \Delta),(k,a_k + \Delta)]$.

Conversely, if $\Delta$ satisfies all of such constraints, $b_i\ge a_i-\Delta$ for all
$1\le i\le n$, then $b_i$ cannot be greater than $a_i+\Delta$ as that would violate $H$ being
the convex lower hull of $(i,a_i + \Delta)$. Thus, $(b_1,\ldots,b_n)$ is a possible
solution.

It takes $O(n)$ time to compute the lower convex hull and $O(n)$ time to calculate the
minimum $\Delta$. Thus, this algorithm solves $L_\infty$ convex regression in $O(n)$ time.
\qed\end{proof}
The above method can also be adapted for inputs with $x$-values that are non-uniformly spaced.
However, it does not directly generalize to weighted $L_\infty$ regression: 
moving points up by $w_i\cdot \Delta$ can lead to different lower hulls for different values
of $\Delta$.

		\clearpage
		\section{Proof of Lemma~\ref{lem:polytope_shift}}\label{subsec:proof}
\label{app:polytope-shift}


The proof of Lemma~\ref{lem:polytope_shift} will be separated into two stages. 
First, we show that the polygon defined by
$\{f_{i}^+(\upperhull(P_{i-1}^{\alpha_i}))\} \cup
\{f_{i}^-(\lowerhull(P_{i-1}^{\alpha_i}))\}$
has upper-hull $\{f_i^-(v_{\mathit{ll}}), f_{i}^+(\upperhull(P_{i-1}^{\alpha_i}))\}$
and lower-hull $\{f_{i}^-(\lowerhull(P_{i-1}^{\alpha_i})), f_i^+(v_{\mathit{ur}})\}$, where
$v_{\mathit{ll}}$ is the first vertex of $\lowerhull(P_{i-1}^{\alpha_i})$ and
$v_{\mathit{ur}}$ is the last vertex of $\upperhull(P_{i-1}^{\alpha_i})$.
Furthermore, this polygon will have slopes between vertices in $[0,1]$.
This property will then allow us to show that $\sndOrdOnlyP $ is equivalent to the convex hull of the vertices, 
which implies the claim.

%

In order to show that the $\sndOrdOnlyP $ has all slopes between $0$ and $1$, 
we consider how $f_i^-$ and $f_i^+$ affect slopes.


\begin{lemma}[Bounded slopes]\label{lem:slope_bounded}
	If $P$ is  $\polytope$, then for any connected vertices $v_j,v_k \in V$, any $i$, and $* \in \{-,+\}$, we have
	\[
	0 \leq \text{slope}(f_{i}^*(v_j),f_{i}^*(v_{k})) \leq 1
	\]
	and for any connected vertices $v_j,v_k,v_l \in V$, 
	if $\text{slope}(v_j,v_k) < \text{slope}(v_k,v_l)$, then 
	\[
	\text{slope}(f_{i}^*(v_j),f_{i}^*(v_{k})) < \text{slope}(f_{i}^*(v_{k}),f_{i}^*(v_{l}))
	\]
\end{lemma}

\begin{proof}
	We first write the slope function explicitly to obtain  
	\[
			\text{slope}(f_{i}^*(v_j),f_{i}^*(v_{k})) 
		\wrel= 
			\frac{(y_j + z_{i}^*) - (y_{k} + z_{i}^*)}{(x_j + y_j + z_{i}^*) - (x_{k} + y_{k} + z_{i}^*)} 
		\wrel=
			\frac{y_j - y_{k}}{(x_j - x_{k}) + (y_j - y_{k})}.
	\]
	This implies that if $\text{slope}(v_j,v_k) = \infty$ then $	\text{slope}(f_{i}^*(v_j),f_{i}^*(v_{k})) = 1$, and if $\text{slope}(v_j,v_k) = 0$ then $	\text{slope}(f_{i}^*(v_j),f_{i}^*(v_{k})) = 0$. Furthermore, this gives the identity
	\[
	\text{slope}(f_{i}^*(v_j),f_{i}^*(v_{k}))^{-1} = 
	\text{slope}(v_j,v_{k})^{-1} + 1
	\]
	when $\text{slope}(v_j,v_{k}) \in (0,\infty).$
	Combined with the fact that all slopes are non-negative,
	this gives both of our desired inequalities.
\qed\end{proof}

The first inequality of the lemma above will allow us to show that all of the slopes
between vertices are bounded, and the second implies that each of the vertices remains a
vertex, giving the following corollary.

\begin{corollary}[Hulls by elementwise transformation]\label{cor:mapping_is_polytope}
~\\
	If $P_{i-1}^{\alpha_i}$ is $\polytope$, then the convex hull $P$ of 
	\( 
		V = f_{i}^+(\upperhull(P_{i-1}^{\alpha_i})) \cup f_{i}^-(\lowerhull(P_{i-1}^{\alpha_i}))
	\)	
	has $\upperhull(P) = \{f_i^-(v_{\mathit{ll}}), f_{i}^+(\upperhull(P_{i-1}^{\alpha_i}))\}$
	and
	$\lowerhull(P) = \{f_{i}^-(\lowerhull(P_{i-1}^{\alpha_i})), f_i^+(v_{\mathit{ur}})\}$,  
	where $v_{\mathit{ll}}$ is the first (lower-left) vertex of $\lowerhull(P_{i-1}^{\alpha_i})$ and $v_{\mathit{ur}}$ is
	the last (upper-right) vertex of $\upperhull(P_{i-1}^{\alpha_i})$. 
	Furthermore, for any connected
	vertices $v_j,v_k$ in $P$, we have \(0 \leq \text{slope}(v_j,v_k) \leq 1\).
\end{corollary}

\begin{proof}
	By construction, the first and last vertices of $\upperhull(P)$ and $\lowerhull(P)$ are
the same. Let $v_{u1}$ be the first vertex of $\upperhull(P_{i-1}^{\alpha_i})$, which
gives two possibilities, either (1): $v_{u1} = v_{\mathit{ll}}$, or (2)
$\text{slope}(v_{u1},v_{\mathit{ll}}) = \infty$. For case (1) it is easy to see that
$\text{slope}(f_i^+(v_{u1}),f_i^-(v_{\mathit{ll}})) = 1$, and for case (2), 
we showed in the proof
of Lemma~\ref{lem:slope_bounded} that $\text{slope}(v_{u1},v_{\mathit{ll}}) = \infty$ implies
$\text{slope}(f_i^+(v_{u1}),f_i^+(v_{\mathit{ll}})) = 1$, which combined with
$\text{slope}(f_i^+(v_{\mathit{ll}}),f_i^-(v_{\mathit{ll}})) = 1$ gives
$\text{slope}(f_i^+(v_{u1}),f_i^-(v_{\mathit{ll}})) = 1$. Furthermore the slopes between all
vertices in $\upperhull(P_{i-1}^{\alpha_i})$ are less than $\infty$ by
Definition~\ref{def:upper}, and therefore less than $1$ under the transformation by
Lemma~\ref{lem:slope_bounded}. Along with the second inequality of
Lemma~\ref{lem:slope_bounded}, this implies that $\upperhull(P)$ makes up a concave
function from $f_i^-(v_{\mathit{ll}})$ to $f_i^+(v_{\mathit{ur}})$.
	
	By symmetric reasoning we see that $\lowerhull(P)$ makes up a convex function from
$f_i^-(v_{\mathit{ll}})$ to $f_i^+(v_{\mathit{ur}})$. Additionally, the second inequality states
that every element in 
$\{f_{i}^-(\lowerhull(P_{i-1}^{\alpha_i})), f_i^+(v_{\mathit{ur}})\}$ and 
$\{f_i^-(v_{\mathit{ll}}), f_{i}^+(\upperhull(P_{i-1}^{\alpha_i}))\}$ must be a vertex. 
Accordingly, $P$ must be a
convex polygon with all slopes between $0$ and $1$.
\qed\end{proof}

We now have fixed upper and lower hulls of a polygon, and we use the representation
as the convex hull its vertices, along with the bounded-slope 
property, to show that this polygon is in fact equal to $\sndOrdOnlyP $. 
In particular, all the slopes being bounded by $1$ will be critical here because each point
$(x,y) \in P_{i-1}^{\alpha_i}$ maps to a line segment from $(x + y + z_i^-,y+z_i^-)$ to
$(x + y + z_i^+,y+z_i^+)$, which has slope $1$. If we then consider $(x,y)$ to be in the
upper hull, if the slopes of our new upper-hull for $\sndOrdOnlyP $ were greater than $1$, 
the point $(x + y + z_i^-,y+z_i^-)$ would lie outside of this hull. 
Our bounded slopes prevent this, though, and lead to the following lemma.

\begin{lemma}\label{lem:mapping_is_desired_polytope}
	Let $P_{i-1}^{\alpha_i}$ be $\polytope$ and let $P$ be the convex hull of
\[
	V \wrel=
	\bigl\{f_{i}^+(\upperhull(P_{i-1}^{\alpha_i}))\bigr\} \cup
	\bigl\{f_{i}^-(\lowerhull(P_{i-1}^{\alpha_i}))\bigr\}.
\]
	Then $P = \sndOrdOnlyP $.
\end{lemma}

\begin{proof}
We show both inclusions.
\begin{itemize}
\item
	$P \subseteq \sndOrdOnlyP $.\\	
	By definition of $P$, any point $u \in P$, can be written as a convex combination 
	\[
	\sum_{(x_j,y_j) \in V(P_{i-1}^{\alpha_i})}p_j\left((x_j+y_j,y_j) + (z_{i}^*,z_{i}^*)\right),
	\] 
	where the sum is over the vertices $(x_j,y_j)$ of $P_{i-1}^{\alpha_i}$, 
	$* \in\{-,+\}$, and $\sum p_j = 1$.
	We set $z = \sum p_j z_{i}^*$; clearly, $z \in [z_{i}^-,z_{i}^+]$. 
	Furthermore set $x = \sum p_j x_i$ and $y = \sum p_j y_j$. 
	We know each $(x_j,y_j)$ is a vertex in
	$P_{i-1}^{\alpha_i}$, so by convexity $(x,y)$ must be in $P_{i-1}^{\alpha_i}$, implying
	$(x + y + x, y+z) \in \sndOrdOnlyP $ by Corollary~\ref{cor:scaling}. 
\item $\sndOrdOnlyP  \subseteq P$.\\	
		Assume towards a contradiction there were 
		$(x+y+z,y+z) \in \sndOrdOnlyP $ with $(x,y) \in P_{i-1}^{\alpha_i}$
	and $z \in [z_{i}^-,z_{i}^+]$, but $(x+y+z,y+z) \notin P$. 
	By definition and assumption,
	both $P$ and $P_{i-1}^{\alpha_i}$ are convex, so there must be a \emph{vertex} $(x_v,y_v)$ of
	$P_{i-1}^{\alpha_i}$ such that $(x_v + y_v + z, y_v + z) \notin P$. Furthermore, by
	convexity of $P$, there must also exist $z \in \{z_{i}^-,z_{i}^+\}$ such that 
	$(x_v + y_v + z, y_v + z) \notin P$.
		Assume \withoutlossofgenerality that $(x_v,y_v) \in \upperhull(P_{i-1}^{\alpha_i})$.
		By definition of $P$, we have
	$(x_v + y_v + z_{i}^+, y_v + z_{i}^+) \in P$, so we must have $z=z_i^-$.
		
	Since $P_{i-1}^{\alpha_i}$ is $\polytope$ and $f_i$ is monotone,
	$f_i^-(v_{\mathit{ll}})$ is dominated%
	\footnote{%
		$(x_1,y_1)$ is said to dominate $(x_2,y_2)$ if $x_1 \ge x_2$ and $y_1\ge y_2$.
	}
	by $(x_v + y_v + z_{i}^-, y_v + z_{i}^-)$, and
	similarly, $f_i^+(v_{\mathit{ur}})$ dominates 
	$(x_v + y_v + z_{i}^-, y_v + z_{i}^-)$. 
	Furthermore, by
	Corollary~\ref{cor:mapping_is_polytope} the upper hull lies above the line segment from
	from $f_i^-(v_{\mathit{ll}})$ to $(x_v + y_v + z_{i}^+, y_v + z_{i}^+)$ and has slope at most 1.
	But the slope between $(x_v + y_v + z_{i}^-, y_v + z_{i}^-)$ and
	$(x_v + y_v + z_{i}^+, y_v + z_{i}^+)$ is exactly $1$, so
	$(x_v + y_v + z_{i}^-, y_v + z_{i}^-)$ cannot lie above the upper hull.
	
	Finally, $(x_v + y_v + z_{i}^-, y_v + z_{i}^-)$ also cannot lie below
	$\lowerhull(P)$ because otherwise there would exist $(x_v,y) \in P_{i-1}^{\alpha_i}$ that
	lies above $(x_v,y_v)$, contradicting $(x_v,y_v)$ being in
	$\upperhull(P_{i-1}^{\alpha_i})$. 
	Because the upper hull and lower hull combine to the
	convex polygon $P$ and because the $x$-coordinate of $(x_v + y_v + z_{i}^-, y_v + z_{i}^-)$ is
	within the range of $x$-coordinate of $P$, we have $(x_v + y_v + z_{i}^-, y_v +
	z_{i}^-)\in P$, a contradiction.
%
\end{itemize}
\qed\end{proof}
With this, we finish the proof of our lemma.
\begin{proof}[of Lemma~\ref{lem:polytope_shift}]
	Follows directly from Corollary~\ref{cor:mapping_is_polytope} 
	and Lemma~\ref{lem:mapping_is_desired_polytope}.
\qed\end{proof}

		\section{Complete algorithm}
\label{app:backtracing}

\RestyleAlgo{boxruled}
\LinesNumbered
\SetAlgoNoEnd
\begin{algorithm}[hbtp]
\small
	\KwIn{Vectors $\vec{x^-} \leq \vec{x^+}$, $\vec{y^-} \leq \vec{y^+}$, $\vec{z^-} \leq \vec{z^+}$, $\vec\alpha \ge 0$}
	
	\KwOut{Some $\vec b\in \solutionset$, or \textit{infeasible} if $\solutionset=\emptyset$.}

	\textsl{\textbf{Note:} We represent vertex $(x,y)$ by real vector
	$(x,y,1)^T$,\quad(homogeneous coordinates).} \\
	
	\textbf{\textsl{[Step 1: Init]}}\\
	$u \gets $ deque with vertices of upper hull of $P_2$ (sorted by $x$-coordinates)\;
	$v \gets $ deque with vertices of lower hull of $P_2$ (sorted by $x$-coordinates)\;
	$S_u \gets I_3$;$\;$ $S_v \gets I_3$ 
		\tcc*{init maps to the identity matrix $I_3$ in $\mathbb R^{3\times3}$}
	
	\textbf{\textsl{[Step 2: Compute $P_i$]}}\\
	\For {$i\gets 3$ \KwTo $n$}{
		$S_u\gets \left(\begin{array}{ccc}1 & \alpha_i & z_i^+ \\ 0 & \alpha_i & z_i^+ \\ 0 & 0 & 1\end{array}\right) \cdot S_u$; \quad
		$S_v\gets \left(\begin{array}{ccc}1 & \alpha_i & z_i^- \\ 0 & \alpha_i & z_i^- \\ 0 & 0 & 1\end{array}\right) \cdot S_v$%
		\tcc*{Update maps}
		\tcc{Add new LL / UR vertex to hulls after transformation}
		$u.\mathit{push\_front}\bigl((S_u)^{-1}\cdot S_v\cdot v.\mathit{front}()\bigr)$;$\;$
		$v.\mathit{push\_back}\bigl((S_v)^{-1}\cdot S_u\cdot u.\mathit{back}()\bigr)$\;
		\For(\tcc*[h]{Cut left and right boundary})
		{$c \in \{u,v\}$}{%
			$r\gets \mathit{null}$;
			\lWhile {$c.\mathit{size}() \ge 1 \wedge (S_c\cdot c.\mathit{front}())_x < x_i^-$}
				{$r\gets S_c\cdot c.\mathit{pop\_front}()$}
			\lIf{$c.\mathit{empty()}$}{\KwRet{\textit{infeasible}}}
			\lIf{$r\ne \mathit{null}$}{%
				$\bigl\{q\gets S_c\cdot c.\mathit{front}()$; 
				$c.\mathit{push\_front}\bigl(
					(S_c)^{-1} \cdot \bigl(q + \smash{\frac{q_x - x_i^-}{q_x - r_x}} \cdot (r - q)\bigr)
				\bigr)\bigr\}$%
			}
			$r\gets \mathit{null}$;
			\lWhile {$c.\mathit{size}() \ge 1 \wedge (S_c\cdot c.\mathit{back}())_x > x_i^+$}
				{$r\gets S_c\cdot c.\mathit{pop\_back}()$}
			\lIf{$c.\mathit{empty()}$}{\KwRet{\textit{infeasible}}}
			\lIf{$r\ne \mathit{null}$}{%
				$\bigl\{q\gets S_c\cdot c.\mathit{front}()$; 
				$c.\mathit{push\_front}\bigl(
					(S_c)^{-1} \cdot \bigl(q + \smash{\frac{x_i^+ - q_x}{r_x - q_x}} \cdot (r - q)\bigr)
				\bigr)\bigr\}$%
			}
		}
		\tcc{Temporarily add vertices for vertical line segments (simplifies cuts)}
		\lIf{$(S_u\cdot u.\mathit{front}())_y > (S_v\cdot v.\mathit{front}())_y$}
			{$u.\mathit{push\_front}\bigl((S_u)^{-1}\cdot S_v\cdot v.\mathit{front}()\bigr)$}
		\lIf{$(S_u\cdot u.\mathit{back}())_y > (S_v\cdot v.\mathit{back}())_y$}
			{$v.\mathit{push\_back}\bigl((S_v)^{-1}\cdot S_u\cdot u.\mathit{back}()\bigr)$}
		\For(\tcc*[h]{Cut upper and lower boundary})
		{$c \in \{u,v\}$}{
			$r\gets \mathit{null}$;
			\lWhile {$c.\mathit{size}() \ge 1 \wedge (S_c\cdot c.\mathit{front}())_y < y_i^-$}
				{$r\gets S_c\cdot c.\mathit{pop\_front}()$}
			\lIf{$c.\mathit{empty()}$}{\KwRet{\textit{infeasible}}}
			\lIf{$r\ne \mathit{null}$}{%
				$\bigl\{q\gets S_c\cdot c.\mathit{front}()$; 
				$c.\mathit{push\_front}\bigl(
					(S_c)^{-1} \cdot \bigl(q + \smash{\frac{q_y - y_i^-}{q_y - r_y}} \cdot (r - q)\bigr)
				\bigr)\bigr\}$%
			}
	%
			$r\gets \mathit{null}$;
			\lWhile {$c.\mathit{size}() \ge 1 \wedge (S_c\cdot c.\mathit{back}())_y > y_i^+$}
				{$r\gets S_c\cdot c.\mathit{pop\_back}()$}
			\lIf{$c.\mathit{empty()}$}{\KwRet{\textit{infeasible}}}
			\lIf{$r\ne \mathit{null}$}{%
				$\bigl\{q\gets S_c\cdot c.\mathit{front}()$; 
				$c.\mathit{push\_front}\bigl(
					(S_c)^{-1} \cdot \bigl(q + \smash{\frac{y_i^+ - q_y}{r_y - q_y}} \cdot (r - q)\bigr)
				\bigr)\bigr\}$%
			}
			\smallskip
			$\mathit{ll}_c \gets S_c\cdot c.\mathit{front}()$;$\;$
			$\mathit{ur}_c \gets S_c\cdot c.\mathit{back}()$%
			\tcc*{Store current LL/UR for later}
			\tcc{Remove generated duplicate nodes and vertical segments}
			\tcc{($\mathit{sndFront/sndBack}$ denote the second / second-to-last elements)}
			\lWhile {$c.\mathit{size}() \ge 2 \wedge (S_c \cdot c.\mathit{front}())_x = (S_c \cdot c.\mathit{sndFront}())_x$}{$c.\mathit{pop\_front}()$}
			\lWhile {$c.\mathit{size}() \ge 2 \wedge (S_c \cdot c.\mathit{back}())_x = (S_c \cdot c.\mathit{sndBack}())_x$}{$c.\mathit{pop\_back}()$}
		}
		\tcc{Add stored LL/UR vertices if horizontal segments missing}
		\lIf{$(S_v\cdot v.\mathit{front}())_x > (S_u\cdot u.\mathit{front}())_x$}
			{$v.\mathit{push\_front}\bigl( (S_v)^{-1}\cdot \mathit{ll}_u  \bigr)$}
		\lIf{$(S_u\cdot u.\mathit{back}())_x < (S_v\cdot v.\mathit{back}())_x$}
			{$u.\mathit{push\_back}\bigl( (S_u)^{-1}\cdot \mathit{ur}_v  \bigr)$}
	}
		
	\textbf{\textsl{[Step 3: Compute $\vec b$]}}\\
	$(x,y)\gets S_u\cdot u.\mathit{back}()$;$\;$ $b_n\gets x$;$\;$
	$p\gets$ index of the last element of $u$\;
	\For{$i \gets n$ \KwTo $3$}{
		Revert $u$, $v$, $S_u$, $S_v$ to the previous stage;\\
		$x'\gets x-y$\;
		\lWhile{$p_x < x'$}{$p$++}
		Use $u_p$ and $u_{p-1}$ (if exists) to compute $y_m \gets \max\{ y'\mid(x',y')\in P_{i-1} \}$\;
		
		\leIf{$y\ge \alpha_i y_m + z_i^-$}{$(x,y)\gets (x',y_m)$}{$(x,y)\gets (x', (y-z_i^-)/\alpha_i)$}
		
		$b_{i-1}\gets  x$\;
	}
	
	$b_1\gets x-y$\;	
	\Return $(b_1,\ldots,b_n)$\;
	
	\caption{1st/2nd-Diff-Constrained Decision Algorithm}
	\label{alg:main}
\end{algorithm}

In this appendix, we give detailed pseudocode for our entire algorithm.
We also discuss the details on the backtracing step,
\ie, computing an actual solution $\vec b\in\solutionset$ from the 
(implicitly represented) feasibility polygons $P_2,\ldots,P_n$.
The final procedure is shown in Algorithm~\ref{alg:main}.

\subsection{Implicitly computing the $P_i$}

The main ideas have been described in Section~\ref{sec:algorithm}.
We represent points in homogeneous coordinates, \ie, 
$(x,y)$ becomes the column vector $(x,y,1)^T$.
That allows our transformation to be represented as a single matrix,
and we can compose them by multiplying the matrices.
We store the current matrix in Algorithm~\ref{alg:main}
in $S_u$ (for the upper hull) and $S_v$ for the lower hull.
$u$ and $v$ denote the deques storing the (untransformed) points 
of $\upperhull$ and $\lowerhull$ in homogeneous coordinates
and in sorted order.

To compute $P_i$ from $P_{i-1}$ (Step 2),
we update the transformation matrices and add the new points
to the hull (following Lemma~\ref{lem:full_mapping}).
After that (line 9), $u$ and $v$ represent $\sndOrdOnlyP$.
To implement the intersection with the half planes corresponding
to the value and first-order constraints at $i$,
we separately cut upper and lower hull with all four boundaries.
Since we store upper and lower hull separately, 
vertical line segments are not explicitly represented in either hull,
which requires some care in cutting with horizontal lines.
We therefore use the following strategy~--it is illustrated on an example in Figure~\ref{fig:example-cuts}: 
We first cut with the 
left and right boundaries (the value constraints),
then transform our representation temporarily to left and right hulls (lines~17--18),
which can easily handle cutting by horizontal line segments.
Cutting is always implemented as a linear scan of $u$ resp.\ $v$, 
during which all vertices outside the constraint halfplane are removed. 
Then we add a new vertex at the intersection
of the last segment with the constraint.
(We remember the last removed vertex $r$ for doing so.)

\begin{figure}[tbh]
	\hspace*{-1em}\begin{tikzpicture}[
		scale=1.5,
		value/.style={blue!80!black},
		first order/.style={green!60!black},
		snd order/.style={red!70!black},
		lower hull/.style={red!50!black},
		upper hull/.style={green!30!black},
	]
	\smaller
	\def\ls{.6}
	\def\r{.4pt}
	
	\begin{scope}[shift={(0.5,0)}]
		\fill[fill=cyan!50] 
				(0,-1) 
				-- (2,1)
				-- (4,2)
				-- (1,-1) 
			-- cycle;
		\draw[lower hull,very thick] 
				(0,-1) circle (\r) node[below left,scale=\ls] {$(0,-1)$}
				-- (2,1) circle (\r) node[above,scale=\ls] {$(2,1)$}
				-- (4,2) circle (\r) node[right,scale=\ls] {$(4,2)$} 
			;
		\draw[upper hull,very thick] 
				(0,-1) circle (\r) node[above right,scale=\ls] {} 
				-- (1,-1) circle (\r) node[below,scale=\ls] {$(1,-1)$} 
				-- (4,2) circle (\r) node[left,scale=\ls] {} 
			;
		\draw[dashed] (1.5, -1.2) -- ++(0,3.25);
		\draw[dashed] (1.95, -1.2) -- ++(0,3.25);
		\node at (1.5,-1.5) {(a)} ;
	\end{scope}
	
	\begin{scope}[shift={(2.5,0)}]
		\fill[fill=cyan!50] 
				(1.5,0.5) 
				-- (1.95,.95)
				-- (1.95,-.05)
				-- (1.5,-.5) 
			-- cycle;
		\draw[lower hull,very thick] 
				(1.5,0.5) circle (\r) node[below left,scale=\ls] {$(1.5,0.5)$}
				-- (1.95,0.95) circle (\r) node[above,scale=\ls] {$(1.95,0.95)$}
			;
		\draw[upper hull,very thick] 
				(1.5,-0.5) circle (\r) node[below,scale=\ls] {$(1.5,-0.5)$} 
				-- (1.95,-0.05) circle (\r) node[right,scale=\ls] {$(1.95,-0.05)$} 
			;
		\node at (1.5,-1) {(b)} ;
	\end{scope}
	
	\begin{scope}[shift={(4,0)}]
		\fill[fill=cyan!50] 
				(1.5,0.5) 
				-- (1.95,.95)
				-- (1.95,-.05)
				-- (1.5,-.5) 
			-- cycle;
		\draw[lower hull,very thick] 
				(1.5,-0.5) circle (\r) node[below left,scale=\ls] {}
				-- (1.5,0.5) circle (\r) node[below left,scale=\ls] {}
				-- (1.95,.95) circle (\r) node[below,scale=\ls] {}
			;
		\draw[upper hull,very thick] 
				(1.5,-0.5) circle (\r) node[above right,scale=\ls] {} 
				-- (1.95,-0.05) circle (\r) node[above,scale=\ls] {} 
				-- (1.95,0.95) circle (\r) node[above,scale=\ls] {} 
			;
		\draw[dashed] (1,-.3) -- ++ (1.4,0) ;
		\draw[dashed] (1,.4) -- ++ (1.4,0) ;
		\node at (1.5,-1) {(c)} ;
	\end{scope}
	
	\begin{scope}[shift={(5.5,0)}]
		\fill[fill=cyan!50] 
				(1.5,-0.3) 
				-- (1.5,0.4)
				-- (1.95,.4) 
				-- (1.95,-0.05) 
				-- (1.7,-.3)
			-- cycle;
		\draw[lower hull,very thick] 
				(1.5,-0.3) circle (\r) node[below left,scale=\ls] {$(1.5,-0.3)$}
				-- (1.5,0.4) circle (\r) node[below left,scale=\ls] {$(1.5,0.4)$}
			;
		\draw[upper hull,very thick] 
				(1.7,-0.3) circle (\r) node[below right,scale=\ls] {$(1.7,-0.3)$} 
				-- (1.95,-0.05) circle (\r) node[right,scale=\ls] {$(1.95,-0.05)$} 
				-- (1.95,0.4) circle (\r) node[above,scale=\ls] {$(1.95,0.4)$} 
			;
		\node at (1.5,-1) {(d)} ;
	\end{scope}
	
	\begin{scope}[shift={(7,0)}]
		\fill[fill=cyan!50] 
				(1.5,-0.3) 
				-- (1.5,0.4)
				-- (1.95,.4) 
				-- (1.95,-0.05) 
				-- (1.7,-.3)
			-- cycle;
		\draw[lower hull,very thick] 
				(1.5,0.4) circle (\r) node[below left,scale=\ls] {}
			;
		\draw[upper hull,very thick] 
				(1.7,-0.3) circle (\r) node[above right,scale=\ls] {} 
				-- (1.95,-0.05) circle (\r) node[above,scale=\ls] {} 
			;
		\node at (1.5,-1) {(e)} ;
	\end{scope}
	
	\begin{scope}[shift={(8.5,0)}]
		\fill[fill=cyan!50] 
				(1.5,-0.3) 
				-- (1.5,0.4)
				-- (1.95,.4) 
				-- (1.95,-0.05) 
				-- (1.7,-.3)
			-- cycle;
		\draw[lower hull,very thick] 
				(1.5,0.4) circle (\r) node[left,scale=\ls] {$(1.5,0.4)$}
				-- (1.95,0.4) circle (\r) node[above,scale=\ls] {$(1.95,0.4)$}
			;
		\draw[upper hull,very thick] 
				(1.5,-0.3) circle (\r) node[left,scale=\ls] {$(1.5,-0.3)$} 
				-- (1.7,-0.3) circle (\r) node[below,scale=\ls] {$(1.7,-0.3)$} 
				-- (1.95,-0.05) circle (\r) node[right,scale=\ls] {$(1.95,-0.05)$} 
			;
		\node at (1.5,-1) {(f)} ;
	\end{scope}
	\end{tikzpicture}
	\caption{%
		Example for lines~10--30 of Algorithm~\ref{alg:main}.
		(a) The polygon $\sndOrdOnlyP$ (after line 9).
		(b) After vertical cuts at $x_i^- = 1.5$ and $x_i^+=1.95$.
		(c) After adding the vertical line segments (line 18).
		(d) After horizontal cuts at $y_i^- = -0.3$ and $y_i^+ = -0.4$ (line 26); 
			$u$ and $v$ represent left and right hull of the correct polygon now,
			but we have to transform them back to upper and lower hull.
			For that, we store the LL and UR vertices.
		(e) After deletion of vertices with same $x$-coordinate (line 28);
			neither vertical, nor horizontal line segments are represented.
		(f) After adding the stored LL and UR vertices (line 30), we obtain the final upper and lower hulls.
	}
	\label{fig:example-cuts}
\end{figure}
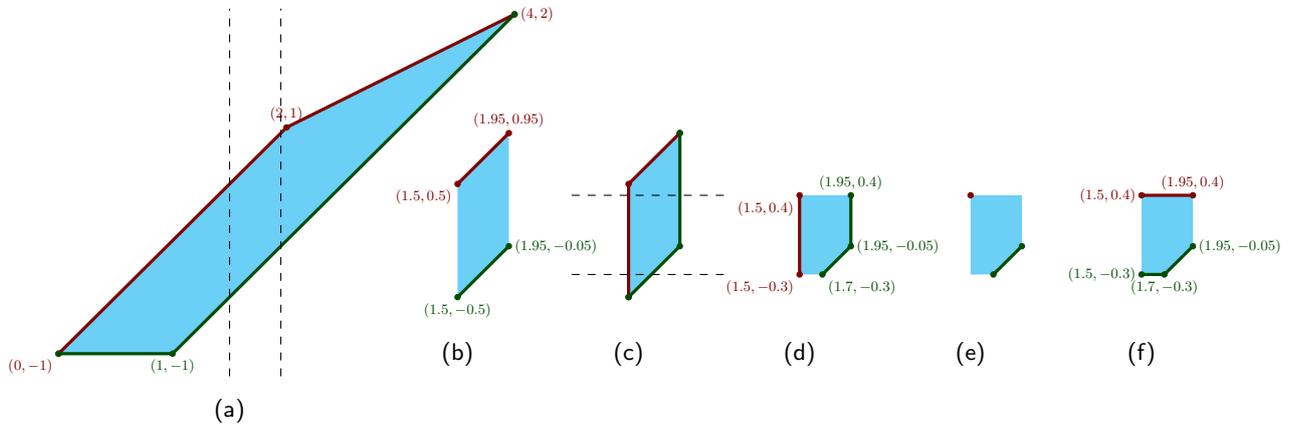

\subsection{Backtracing}

Suppose we have computed $P_n$ as described above,
and then partially backtraced through a sequence of feasible points. 
We are now at $(b_{i+1},b_{i+1}- b_i)$ in
$P_{i+1}$. Since $(b_{i+1}, b_{i+1}-b_i)=(x+y+z,y+z)$, $z\in [z_{i+1}^-,z_{i+1}^+]$ for
some (unknown) $(x,y)=(b_{i}, \alpha_{i+1}(b_{i}-b_{i-1}))\in P_i$, we can recover $x =
b_i$ from $(b_{i+1}, b_{i+1}-b_i)$ by subtracting the two coordinates of $(b_{i+1},
b_{i+1}-b_i)$. To recover $y$, suppose we can find $y_{\max}=\max\{y\mid (b_i,y)\in
P_{i}^{\alpha_{i+1}}\}$ efficiently. 
Since $\{y\mid (b_i,y)\in P_{i}^{\alpha_{i+1}}\}$ is
an interval, the following lemma allow us to find $b_i - b_{i-1}$.%
\begin{lemma}[back 1 step]
\label{lem:back_1_step}
	Let $f_{i+1}(x,y)=\{(x+y+z,y+z)\mid z\in [z_{i+1}^-,z_{i+1}^+]\}$. Either $(b_{i+1},
	b_{i+1}-b_i)\in f_{i+1}((b_i, y_{\max}))$ or $(b_{i+1},
	b_{i+1}-b_i)=(b_i+y+z_{i+1}^-,y+z_{i+1}^-)$ for some $y<y_{\max}$.
\end{lemma}
Intuitively, a vertical line segment $L$ inside $P_{i}$ is mapped to a line-segment with
slope $1$ in $P_{i+1}$, because the line segments the points in $L$ are mapped to
lie all on the same line (overlapping with each other).
\begin{proof}
If $(b_{i+1}, b_{i+1}-b_i)\not\in f_{i+1}(b_i, y_{\max})$, by the maximality of $y_{\max}$,
$b_{i+1}-b_i<y_{\max}+z_{i+1}^-$. Since there exists $(b_i,y')$ such that $(b_{i+1},
b_{i+1}-b_i)\in f_{i+1}(b_i,y')$,  $(b_i+y'+z,y'+z)=(b_{i+1},b_{i+1}-b_i)$ for some $z\in
[z_{i+1}^-,z_{i+1}^+]$. Consider $f_{i+1}(b_i,y+z-z_{i+1}^-)$. Then $(b_{i+1},
b_{i+1}-b_i)=(b_i+(y'+z-z_{i+1}^-)+z_{i+1}^-,(y'+z-z_{i+1}^-)+z_{i+1}^-)$. Since
$b_{i+1}-b_i<y_{\max}+z_{i+1}^-$, $y'+z-z_{i+1}^-<y_{\max}$. The lemma is proven by letting
$y$ be $y'+z-z_{i+1}^-$.
\qed\end{proof}
In the former case of Lemma \ref{lem:back_1_step}, we can take $(x,y_{\max})$ as $(b_{i},
\alpha_{i+1}(b_{i}-b_{i-1}))$. In the latter case, we can take $(b_i,(
b_{i+1}-b_i)-z_{i+1}^-)$ as $(b_{i}, \alpha_{i+1}(b_{i}-b_{i-1}))$.

Since $y_{\max}$ is the $y$-coordinate of the intersection of $\upperhull(P_{i})$ and the
vertical line $(b_i,\cdot)$, to compute $y_{\max}$, we want to find two vertices in
$\upperhull(P_i)$, $(x_l,y_l)$ and $(x_r,y_r)$, such that $x_l\le b_i\le x_r$. $(b_i, y_{\max})$ is just the intersection of the line segment 
between $(x_l,y_l)$ and $(x_r,y_r)$ and the vertical line $(b_i,\cdot)$. The following lemma shows how to find
$(x_l,y_l)$ and $(x_r,y_r)$ efficiently using an amortized constant-time algorithm. 
\begin{lemma}[Computing $y_{\max}$]
\label{lem:retrieve}
Suppose  $(x_l,y_l)$ and $(x_r,y_r)$ are two vertices in $\upperhull(P_i)$, and some point
$(b_i,y)\in P_i$ satisfies $x_l\le b_i\le x_r$. Let $(x',y')$ be some point in $P_{i+1}$
with $(x',y')\in f_{i+1}(b_i,\alpha_{i+1}y)$. Then $x'\le
(f_{i+1}^+(x_r,\alpha_{i+1}y_r))_x$, where $(\cdot)_x$ means taking the $x$-coordinate of
a point and $(\cdot)_y$ takes the $y$-coordinate.
\end{lemma}
\FullVersion{
\begin{proof}
Assume towards a contradiction that $x'>(f_{i+1}^+(x_r,\alpha_{i+1}y_r))_x$. 
Since $x'-y' =b_i \le x_r=(f_{i+1}^+(x_r,\alpha_{i+1}y_r))_x - (f_{i+1}^+(x_r,\alpha_{i+1}y_r))_y$,
we have 
$y' > (f_{i+1}^+(x_r,\alpha_{i+1}y_r))_y$. But $y'=\alpha_{i+1}y+k\le
\alpha_{i+1}y+z_{i+1}^+\le \alpha_{i+1}y_r+z_{i+1}^+= (f_{i+1}^+(x_r,\alpha_{i+1}y_r))_y$.
Contradiction.
\qed\end{proof}
}
The amortized constant-time algorithm to retrieve $(b_n,\ldots,b_1)$ depends on the
implementation of the deques. 
Since we will add $n$ vertices to the deques during the
whole algorithm, the (textbook) fixed-size array-based implementation suffices;
we recall it to fix notation.
A deque $d$ is represented by array $A$ and two indices
$p_l$, $p_r$. $p_l$ is the index of the first element of $d$ and $p_r$ is the index of the
last element. If we want to add an element $e$ to the left of the deque, the two
operations $p_l\gets p_l-1$, $A[p_l] = e$ suffice. Similarly, we can add/pop
elements from left/right. During our algorithm, $p_l$ (resp. $p_r$) can move to the left
(resp.\ right) by at most $n$ positions, so $A$ can be an array of length $2n+O(1)$. If we
store the vertices of $P_2$ in the middle of $A$ initially, we never exceed the boundaries
of $A$ when running the algorithm.

\begin{definition}[Position]
\label{def:pos}
We define $\mathit{pos}_i(x')$ as the smallest index (in the array representing deque $u$) 
of a vertex of $\upperhull(P_i(\cdot))$
with $x$-coordinate at least $x'$.
\end{definition}
Note that adding or removing elements does not change the vertex at a given index 
(unless that vertex itself is removed).
\begin{lemma}[Monotonicity of positions]
\label{lem:retrieve_amortization}~\\
$\mathit{pos}_{i}(b_i) \ge \mathit{pos}_{i+1}(x')$ for some $(x',y')\in f_{i+1}(b_i,\alpha_{i+1}y)$.
\end{lemma}

\FullVersion{
\begin{proof}
By Lemma \ref{lem:retrieve}, $x'\le (f_{i+1}^+(x_r,\alpha_{i+1}y_r))_x$. So
$f_{i+1}^+(x_r,\alpha_{i+1}y_r)$ is stored after $pos_{i+1}(x')$. And since our algorithm
stores $f_{i+1}^+(x_r,\alpha_{i+1}y_r)$ at the same place as $(x_r,y_r)$,
$pos_{i+1}(x')\le pos_{i}(b_i)$.
\qed\end{proof}
}
Lemma~\ref{lem:retrieve_amortization} allows us to find $\mathit{pos}_i(z)$ by moving a pointer
monotonically to the right. 
Thus, we can retrieve $b_n,\ldots,b_1$ in order by unrolling our linear
algorithm for the decision problem and moving the pointer $\mathit{pos}_i(z)$. 
This process takes $O(n)$ time overall.

\subsection{Analysis}

We conclude with the proof of our main theorem.

\begin{proof}[of Theorem \ref{thm:main}]
The correctness of Algorithm~\ref{alg:main} follows from the preceding discussions:
By Lemma~\ref{lem:full_mapping}, the iterative transformations compute the $P_i$ as defined
in~\eqref{def:each_polytope}, and $\solutionset \ne \emptyset$ iff $P_n\ne \emptyset$.
Moreover, Lemma~\ref{lem:back_1_step} shows that, when $\solutionset \ne \emptyset$,
Step 3 computes a valid $\vec b\in \solutionset$.
It remains to analyze the running time.
\begin{itemize}
\item Step~1 
	takes $O(1)$ time since the vertices of $P_2$ are a subset of the (at most) 12 intersection
	points of the defining lines.
	($P_2$ is the trapezoid spanned by
	$(x_2^-,x_2^- - x_1^+),(x_2^-,x_2^- - x_1^-),(x_2^+,x_2^+ - x_1^+),(x_2^+,x_2^+ - x_1^-)$,
	intersected with the halfspaces $y \ge y_2^-$ and $y \le y_2^+$.)
\item Step 2.
	The operations inside the loops are all constant-time and the outer loop runs $O(n)$ times.
	Moreover, the inner while-loops all remove a node from a deque, so their
	total cost over all iterations of the for-loop is $O(n)$, too: 
	We start with $O(1)$ vertices and adding at most $O(n)$ vertices
	throughout the entire procedure (Lemma~\ref{lem:invariant}), so 
	we cannot remove more than $O(n)$ vertices.
\item Step~3. 
	All operations except for the first line inside the for-loop take constant time.
	The inner while-loop runs for overall $O(n)$ iterations, since
	$p$ only moves right and we add $O(n)$ vertices in total. 
	
	It remains to implement the first line of the loop body in $O(n)$ overall time.
	To be able to undo the changes to $u$, $v$, $S_u$, $S_v$, we keep a
	\emph{log} for each instruction executed in Step 2, so that we can undo their changes here 
	(in the opposite order).
	Since Step 2 runs in $O(n)$ total time, the rollback also runs in $O(n)$ time.
\end{itemize}
Since all three steps run in linear time, so does the whole algorithm.
\qed\end{proof}

\FloatBarrier

		\section{Generalization to DAGs is hard}
\label{app:dags-hard}

In this appendix, we will give a natural generalization of Definition~\ref{def:problem} 
to arbitrary DAGs and investigate its complexity.
Our original setting with differences of adjacent indices only
corresponds to a directed-path graph.

In light of rather general results for isotonic regression, the path setting 
might appear quite restrictive; we will argue here why
these conditions probably cannot be relaxed much further if we want an 
$O(n)$ time algorithm.

\begin{definition}\label{def:dagproblem}
	Suppose we are given a directed acyclic graph $G = (V,E)$ with $m = |E|$ edges and
$m_{\pa}$ number of length two directed paths in $G$ , $n$-dimensional vectors $x^- \leq
x^+$, $m$ dimensional vector $y^- \leq y^+$,   and $m_{\pa}$ dimensional vectors $z^- \leq
z^+$ and $\alpha \ge 0$.  We define $\solutionset_G$ to be the set of all $n$-dimensional
vectors $b$ such that $x_i^- \leq b_i \leq x_i^+$ for all $i$,  $y_{ij}^- \leq b_j - b_{i}
\leq y_{ij}^+$ for all edges $(i,j) \in E$, and  $z_{ijk}^- \leq (b_k - b_{j}) -
\alpha_{ijk}(b_{j} - b_{i}) \leq z_{ijk}^+$ for all pairs of edges $(i,j), (j,k) \in E.$
\end{definition}	

In contrast to Theorem~\ref{thm:main}, we show that determining if $\solutionset_G$ if
empty or not is as hard as solving linear programs.

\begin{theorem}
\label{thm:lowerbound}
With notation as in Definition~\ref{def:dagproblem}, if we can determine $\solutionset_G$
is empty or not in time $f(n+m+m_{\pa})$, then we can determine feasibility of any set of
linear constraints defined by $s$ bounded integer coefficients in $c_1f(c_2s\log M))$
time, where $c_1$ and $c_2$ are two constants and the absolute value of each coefficient
in the linear constraints is no more than $M$.
\end{theorem}

Our reduction to prove Theorem~\ref{thm:lowerbound}  is closely motivated by the hardness of isotropic
total variation from~\cite{Kyng2017}, as well as subsequent
works on extending such hardness results to positive linear programs.
Compared to these results though, it sidesteps linear systems, and
is a more direct invocation of the completeness of 2-commodity flow
linear programs from~\cite{Itai1978}.

 We first consider a more restricted class of problems than
Definition~\ref{def:dagproblem} allows (where all the $\alpha$'s in
Definition~\ref{def:dagproblem} are set to be $1$). Formally we define the problem as:
\begin{definition}
\label{def:dagproblem2}
A generalized 
second-order constrained feasibility problem
is defined by variables $b_{1} \ldots b_n$, combined
with a set of $m$ constraints parameterized by
\begin{enumerate}
\item Upper and lower bounds on the variables $x_i^{-}$ and $x_i^{+}.$ 
\item Upper and lower bounds on the first order differences  $y_i^{-}$ and $y_i^{+}$ and corresponding indices $p_{i} < q_{i}.$
\item Upper and lower bounds on the second order differences $z_i^{-}$ and $z_{i}^{+}$ and corresponding indices $r_{i} < s_{i} < t_{i}$
\end{enumerate}
and constraints
\begin{description}
\item[Value Constraints:] $ x_{i}^{-} \leq b_i \leq x_{i}^{+}$ 
\item[First Order Constraints:]   $y_i^{-} \leq b_{q_{i}} - b_{p_{i}}  \leq y_i^{+}$
\item[Second Order Constraints:] 
$z_{i}^{-}
\leq
\left( b_{t_{i}} - b_{s_{i}} \right)
-
\left( b_{s_{i}} - b_{r_{i}} \right) \leq
z_{i}^{+}. $
\end{description}
The goal is to decide whether there exists
$b_{1}, \ldots, b_{n}$ that satisfy all these constraints
simultaneously.
\end{definition}

\begin{proof}[of Theorem \ref{thm:lowerbound}]
It is easy to see that the problem defined in Definition~\ref{def:dagproblem2} is a
special case of the problem in Definition~\ref{def:dagproblem}. This is obtained by
forming a DAG with edges $(p_i,q_i)$, $(r_i,s_i)$, $(s_i,t_i)$ for all 
$p_i$, $q_i$, $r_i$, $s_i$, $t_i$.  
We will prove that a general linear programming feasibility problem with $s$ 
polynomially-bounded integer coefficients can be expressed as a second-order-constrained feasibility
problem (Definition~\ref{def:dagproblem2}). In particular, we will show that a feasibility
of a set of linear constraints containing at most $s$ non-zero coefficients whose absolute
values are integers no more than $M$ can be reduced to $O(s \log M)$ value, first order
and second order constraints as in Definition~\ref{def:dagproblem2}.

Note that the second constraint in Definition~\ref{def:dagproblem2} is the same as
\[
z_{i}^{-}
\leq
2 b_{q_{i}} - b_{r_{i}} - b_{p_{i}}
\leq z_{i}^{+}.
\]

In particular, it allows us to create constraints of the form
\[
2 b_{q_{i}} = b_{p_{i}} + b_{r_{i}}.
\]
We will now show how we can restate a feasibility of a set of general linear constraints
can be expressed as a second order constrained feasibility problem as in
Definition~\ref{def:dagproblem}.
 The main idea will be clear when we consider a linear constraint of the form 
\[
b_{i_1} + b_{i_2} + \ldots b_{i_{k}}
\leq
c_{i},
\]
with $k$ a power of $2$, and
$i_1 < i_2 < \ldots < i_k$ in increasing order. To express this in terms of second order constraints, 
 we can introduce new variables
\begin{align*}
i_1 &< i_{12} < i_{2}\\
i_3 &< i_{34} < i_{4}\\
\ldots
\end{align*}
and use $b_{i_{12}}$ to represent the sum of $b_{i_1}$
and $b_{i_2}$ and so on.
Repeating this halves the value of $k$,
but aggregates the whole sum into a single variable. Therefore, we can express the above
linear constraint as one value constraint
\[
b_{i_{12\ldots k}}
\leq
x^{+}_{12\ldots k}
\vcentcolon=
c_i
\]
and $k-1$ second order constraints  
\[b_{i_{12}} = b_{i_1} + b_{i_2}, \ldots,b_{i_{(k-1)k}} =  b_{i_{k-1}} + b_{i_{k}} ,
\ldots, b_{i_{1\ldots k}} = b_{i_{1\ldots k/2}} + b_{i_{k/2+1\ldots k}} .\] 
In case $k$ is not a power of $2$,
we can add dummy variables whose values we restrict to zero
using the value constraints. This process uses at most $k$ value constraints. So we have
shown that we can express any linear constraint of the form $b_{i_1} + b_{i_2} + \ldots
b_{i_{k}} \leq c_{i}$ in terms of $O(k)$ second order constraints and $O(k)$ value
constraints.

Now consider the case with both positive and negative values in the linear constraint 
$$b_{i_1} \pm \ldots \pm b_{i_k} \leq c_i.$$
We can aggregate the sums of the variables with positive coefficients and negative
coefficients separately, and let us denote the resulting variables by
$b_{\text{pos}},b_{\text{neg}}.$
We can now bound the difference using a first order constraint of the form
\[
b_{\text{pos}} - b_{\text{neg}}
\leq
c_i.
\]
This results in additional $O(1)$ first order constraints for each linear constraint. 

Finally, when the coefficients are arbitrary integers, we can do pairing based on the
binary representation. The second order constraint and value constraint allows us to
create constrains of the form 
\[
0\le 2b_i-b_0-b_j\le 0
\]\[
0\le b_0\le 0
\] which are equivalent to 
\[
b_j=2b_i.
\] 
So we can introduce new variables $d_{kj}$ representing $2^kb_j$ for any $1\le k\le c$
where $c$ is a constant. Thus, given any linear constraint in $k$ variables with integer
coefficients that are bounded by $M$, we first represent each coefficient by its binary
representation, increasing the number of non-zero coefficients by $O(\log M)$ times and
creating $O(k\log M)$ second order constraints and value constraints. Then all the
coefficients in the linear constraints are $+1$ or $-1$ and we can use the reduction
above. In summary, we can solve any linear programming feasibility problem with $O(s)$
non-zero coefficients which are integers bounded by $M$ by a generalized second-order
constrained feasibility problem of $O(s\log M)$ constraints. This together with our
assumption of an algorithm solving generalized second-order constrained feasibility
problem in $f(\cdot)$ time prove the theorem.
\qed\end{proof}

	\end{appendix}

\end{document}